\documentclass[submission,copyright,creativecommons]{eptcs}

\usepackage[T1]{fontenc}
\usepackage[english]{babel}
\usepackage[utf8]{inputenc}
\usepackage{amsmath}
\usepackage{amsthm}
\usepackage{amsfonts}
\usepackage{cleveref}

\usepackage{color}
\usepackage{MnSymbol}
\usepackage{wasysym}
\usepackage{graphicx}
\usepackage{tikz}
\usepackage{setspace}
\usepackage{float}

\newtheorem{definition}{Definition}
\newtheorem{theorem}{Theorem}
\newtheorem{lemma}{Lemma}
\newtheorem{corollary}{Corollary}

\newtheorem{notation}{Notation}
\newtheorem{example}{Example}

\newcommand{\so}{\rightarrow}
\newcommand{\mso}{\Rightarrow}
\newcommand{\miff}{\Longleftrightarrow}

\newcommand{\pair}[1]{\langle #1 \rangle}

\newcommand{\Fcal}{\mathcal{F}}

\newcommand{\abr}{\allowbreak} 

\usepackage[backend=bibtex,style=alphabetic,maxbibnames=15,maxcitenames=6,dateabbrev=false]{biblatex}
\title{Characterization of Lattice Properties\\ Within Modal Extensions}
\author{Alfredo R. Freire
\institute{Department of Philosophy\\University of Brasilia\\ Brasilia, Brasil}
\email{alfredo.filho@unb.br}
\and
  Manuel A. Martins
\institute{Department of Mathematics\\ University of Aveiro\\ Aveiro, Portugal}
\email{\quad martins@ua.pt}
}

\def\titlerunning{Characterization of Lattice Properties Within Modal Extensions}
\def\authorrunning{A. R. Freire \& M. A. Martins}

\begin{document}

\maketitle

\begin{abstract}

This paper investigates the extension of lattice-based logics into modal languages. We observe that such extensions admit multiple approaches, as the interpretation of the necessity operator is not uniquely determined by the underlying lattice structure. The most natural interpretation defines necessity as the meet of the truth values of a formula across all accessible worlds---an approach we refer to as the \textit{normal interpretation}. We examine the logical properties that emerge under this and other interpretations, including the conditions under which the resulting modal logic satisfies the axiom $K$ and other common modal validities. Furthermore, we consider cases in which necessity is attributed exclusively to formulas that hold in all accessible worlds. 

\end{abstract}

\section{Introduction}

Introduced by Kripke \cite{Kripke1959-KRIACT}, the classical modal semantics consists of a set of worlds, an accessibility relation between worlds, and valuations of propositional variables over the two-valued Boolean algebra $B_2$.
However, instead of using $B_2$ as the basis for valuations and non-modal connectives, one may start from an alternative lattice semantics such as the \L ukasiewicz algebra $[0, 1]$, the four-valued Belnap lattice or the Boolean algebras of $n$ with $n > 2$. In this paper, we study modal extensions of lattice semantics, how some modal sentences can impose constraints on the lattices.

Lattices can be used as semantics for an important family of logics including classical, intuitionistic, fuzzy, relevant, and paraconsistent systems. They interpret disjunction and conjunction in a familiar way and structure logical values in a partially ordered set. 
We shall briefly discuss possibilities for building modal extensions from a given lattice semantics. However, as the ordered structure of lattices suggests natural choices for interpreting the necessity operator -- namely, the greatest lower bound of values in the accessed worlds -- we adopt this interpretation for most of our results. 

Our interest in this kind of modal logic is rooted in the problem of meaningfully interpreting, within a given logic, the truth values of another logic. This is fundamentally important if we want to address the \textbf{counterlogical} statements.\footnote{The issue of counterlogical statements has been treated in recent literature by Nolan in \cite{
nolanImpossible}, Berto and Jago in the recent book \cite{berto2019impossible} and Kocurek and Jerzak in \cite{kocurek2021counterlogicals}. The general approach is to consider, in opposition to \textbf{normal} worlds, the \textbf{impossible} worlds where formulas do not always receive values by algebraic operations from the values of propositional variables. This, we believe, makes the study of logical possibilities very restrictive. 
For example, one cannot access an intuitionistic world from the point of view of a classical world, and one must choose a logic that should be considered normal.} Moreover, this issue has appeared in applications to computer systems where states operate using different algebraic properties (see \cite{MMB18, BMC14, DS07}). In this context, a lattice $L$ containing lattice semantics $L_1$ and $L_2$ as sublattices provides us with a common order in which logics based on $L_1$ and $L_2$ can meaningfully communicate. This framework was introduced in \cite{RM24}, and a general mathematical treatment of these structures is worth developing. Here, we limit our work to structures where every world operates in a single lattice semantics. This is a fundamental and required development if we want to properly generalize the treatment of many logics across many worlds proposed in \cite{RM24} (see also \cite{FMM24} for an institutional formulation of Many Logics Modal Logic). 

After a short section on lattice semantics background and a section on modal extension of lattice semantics, we study the relation between general modal validities and lattice properties. We will first observe that factorization of $\Box$ over disjunction is fundamentally related to the order relations in the set of designated values. In addition, we will observe some general conditions connected to the validity of axiom K. The paper concludes with some lines for future research.

\section{Background on lattice semantics}

Before going into the content of our paper, let us briefly introduce some definitions and basic concepts that we will use throughout the paper. 

A \textbf{partially ordered set} is an ordered pair $(P,\leq)$ in which $P$ is a set and $\leq$ is a binary relation on $P$ reflexive, anti-symmetric and transitive. Two elements $a,b\in P$ are said to be \textbf{ comparable to $\leq$} if $a\leq b$ or $b\leq a$. The order $(P,\leq)$ will be called \textbf{linearly ordered} if any two elements can be compared using $\leq$. 

A {\em lattice} is a partially ordered nonempty set $(L,\leq)$ in which each pair of elements $a$, $b$ has a join (denoted by $a+b$) and a meet (denoted by $a.b$). In some contexts, it is useful to consider a unary operation on $L$ called \textbf{complementation} and denoted by the symbol $-$\footnote{The name complementation suggests that it will be used in the interpretation of negation and suggests that it will have the Boolean properties $-a + a = 1$ and $-a.a = 0$. Although we intend to use complementation to interpret negation, we shall not commit ourselves upfront to any property for the operation of complementation.}. 
A lattice $L$ with complementation is \textbf{anti-monotone} when for every $a,b \in L$, if $a \leq b$, then $-b \leq -a$.
Let $\bigwedge S$ denote (if it exists) the greatest lower bound of the values in a set $S$, we say that the lattice $L$ has \textbf{down-distribution} when, for all $A \subseteq L$ and $B \subseteq L$, we have $\bigwedge (A + B) = \bigwedge A + \bigwedge B$. In general, we will use $a \supset b$ to refer to $-a + b$, although some alternatives will also be considered.

In this paper, a \textbf{set of designated elements} of a partially ordered set $\pair{L,\leq}$ is a subset $A$ of $L$ that is closed upward\footnote{This, of course, diminishes the generality of the set of designated values being considered. However, it is rare to see such sets employed in defining the concept of validity within a logical framework.}, that is, for any $a,b\in L$, $a\leq b$ and $a\in A$ imply $b\in A$. 
For a lattice $L$, we say that $A$ is \textbf{implicative} in $L$ when 
$$a \leq b \mso a \supset b \in A.$$
A set of designated elements of $L$ is called a \textbf{filter} if it is closed under meet operation. A filter is called an \textbf{ultrafilter} if it is maximal with respect to inclusion. 

A \textbf{matrix} is a pair $\pair{L,D}$, where $L$ is a lattice and $D$ is a set of designated values. If $D$ is a filter, $\pair{L,D}$ is called a \textbf{filter-matrix}.

Let $\mathcal{L} = \{Var, \land, \lor, \lnot\}$ be a propositional language and $Var$ be a set  of propositional variables; $Form(\mathcal  L,\abr Var)$, the set of propositional formulas of $\mathcal{L}$, is the free $\mathcal{L} $-algebra over $Var$. A \textbf{valuation} is a function of $Var$ in $L$. It is well known that it can be uniquely extended to a homomorphism, namely the map $\bar v:Form(\mathcal L,Var)\to L$ defined in the following way

	\begin{enumerate}
		\item $\bar v(p)=v(p)$ for every variable $p\in Var$;
		\item $\bar v(\varphi \lor \psi) = \bar v(\varphi) + \bar v ( \psi)$;
		\item $\bar v(\varphi \land \psi) = \bar  v(\varphi) . \bar v ( \psi)$;
		\item $\bar v (\lnot \varphi) = - \bar v(\varphi)$.
		\end{enumerate}
Since $\bar v$ is uniquely determined by $v$ we simply write $v$. A \textbf{substitution} is an automorphism in the formula algebra $Form(\mathcal L,Var)$.
Given a matrix $M=(L,D)$, a valuation $v$ and a formula $\varphi$, we write $v\vDash_M \varphi$ whenever $v(\varphi)\in D$. The \textbf{consequence relation over $M$} is the relation $\vDash_M\subseteq P(Form(\mathcal L,Var))\times Form(\mathcal L,Var)$ defined by 
 \begin{center}
$\Gamma\vDash_M \varphi $ iff for all valuations $v$, $v(\Gamma)\subseteq D$ implies $v(\varphi)\in D$.
\end{center}

When it is clear from the context, we omit explicit references to $M$ in $\vDash_M$. The relation $\vDash_M$ satisfies Tarski's conditions.

\section{Modal extension of lattice based logics}

We restrict our investigation to extensions of logics that have a lattice semantics as suggested in \cite{fitting1991many, priest_2008, odintsov2012bk, RM24}. This restriction will allow us to produce a broader study of the phenomena of modal structures since it produces common meet and join operations shared by most familiar logics.

As we shall see shortly, there are many alternatives for interpreting the validity of $\Box \varphi$ in a world. If we consider a model $M$ where every world operates in a three-valued lattice $\{0, 0.5, 1\}$ with distinguished values $\{0.5, 1\}$, one could say that the value of the formula $\Box \psi$ is $1$ whenever all accessible worlds validate $\psi$ (i.e., $\psi$ has a designated value in all accessible worlds). We could, however, attribute the value $0.5$ or find other more articulated ways of attributing the truth value for the modal formula. We shall call a modal valuation \textit{regular} when it implies that ``necessity means true in all accessible worlds'' (NAW). More precisely, let us define a general modal extension of these logics:

\begin{definition}
Let $\mathcal{L}$ be a propositional language, and $Var$ be the set of propositional variables of $\mathcal{L}$.
For a kripke frame\footnote{In modal logic a \textbf{Kripke frame} $F$ is a pair$ = \pair{W, R}$, where $W$ is a set (of worlds) and $R\subseteq W^2$ (the accessibility relation).} $F = \pair{W, R}$ and a lattice $A$ , we say that $v$ is an \textbf{$A$-valuation} of $F$ in $\mathcal{L}$ when $v$ is a function from $W \times Var$ to $A$.

The function $v_w: Var \longrightarrow A$ is an $A$-valuation such that $v_w(X) = v(w,X)$ for every variable $X \in Var$.
\end{definition}

\begin{definition}
Let $F = \pair{W, R}$ be a Kripke frame and $v$ an $A$-valuation for a lattice $A$, we define $M = \pair{W, R, v}$ to be a \textbf{$A$-model} with frame $F$.
\end{definition}

\begin{definition}
For a lattice $A$ with designated values $D$ and corresponding logic semantics $\vDash_{A, D}$, we say that $\Vdash$ is a  \textbf{modal extension} of $\vDash_{A, D}$ when for every $A$-model $M = \pair{W, R, v}$ and every world $w \in W$ we have 
\begin{enumerate}
	\item[1.] for every formula $\varphi$ without modal operator and $w \in W$,
	$$ w \Vdash \varphi \miff v_w \vDash_{A, D} \varphi;$$
\end{enumerate}
We will say that the modal extension $\Vdash$ is \textbf{regular} when additionally 
\begin{enumerate}
	\item[2.] for every formula $\psi$ and every $w \in W$,
	$$w \Vdash \Box \psi \miff \text{for every } w' \in W \text{ with } w R w', w' \Vdash_{A, D} \psi$$
\end{enumerate}
\end{definition}

\begin{example} Let $\mathbf{B}_2$ be the Boolean algebra with two elements $0$ and $1$. Consider the consequence relation over the matrix $({\mathbf{B}_2, \{1\}})$ (i.e., $\vDash_{B_2, \{1\}}$).

The modal extension of $\vDash_{B_2, \{1\}}$ where 
$$w \Vdash \Box \psi := \bigwedge\{ w' \Vdash \psi|w R w'    \}$$

\noindent coincides with the standard modal logic which is a regular modal extension of $\vDash_{B_2, \{1\}}$.

On the other hand, taken the same $\vDash_{B_2, \{1\}}$, and considering the modal extension with $$w \Vdash \Box \psi := w \Vdash \psi $$ \noindent we obtain a non-regular modal extension.

We also have non-trivial examples of non-regular modal extensions.

Consider the logic $\vDash_{B, U}$, where $B$ is an infinite Boolean algebra and $U$  is a non-principal ultrafilter on $B$.

Define

$$w \Vdash \Box \psi := \bigwedge\{ w' \Vdash \psi|w R w'    \}.$$

Consider a model $M$ with a world $w_1$ that accesses all $w_b$ for each $b \in U$.
If we take $v_{w_a}({p}) = a$,  the value of ${p}$ is designated in each world. i.e. $w_a \Vdash p$. However,  $w_1\Vdash \Box {p}=\bigwedge U \notin U$ as $U$ is non-principal.
Consequently, $w_1 \nvDash \Box {p}$.
\end{example}

Naturally, if $A$ has more than one distinguished value, it is easy to show that there are more than one regular modal extension. Moreover, whenever the lattice is complete, 
the most natural modal extension assigns to $\Box \varphi$ in a world $w$ the conjunction of the values of $\varphi$ in the accessible worlds. Thus, we define:

\begin{definition}\label{normal-modal-valuation}
Let $\mathcal{L} = \{Var, \land, \lor, \lnot\}$ be a propositional language, and $Form$ be the set of propositional formulas of $\mathcal{L}$.
For a  \textbf{complete} lattice $A = \pair{D, ., +, -}$ where $.$ is the usual meet operation, $+$ the usual join operation, and $-$ is any unary operation over $D$, consider the frame $F = \pair{W, R}$. We say that $v: W \times Form \longrightarrow A$ is a  \textbf{normal modal valuation} when for every $w \in W$
\begin{enumerate}
	\item $v(w, X) \in A$ for every variable $X$ of $\mathcal{L}$.
	\item $v(w, \varphi \lor \psi) = v(w, \varphi) + v(w, \psi)$.
	\item $v(w, \varphi \land \psi) = v(w, \varphi) . v(w, \psi)$.
	\item $v(w, \lnot \varphi) = - v(w, \varphi)$.
	\item $v(w, \Box \varphi) = \bigwedge \{v(w', \varphi) \mid w R w'\}$.
\end{enumerate}
We generally use the notation $v_w(\varphi)$ for $v(w, \varphi)$.
\end{definition}

Note that unless the lattice is complete, the definition of modal valuations will be partial. 
In cases where the base lattice is incomplete, there would be modal formulas that cannot have a truth value. This is because there will be a set of values for which the meet operation does not have a determinate value.
We will return briefly to this topic later in this paper.

\begin{definition}
    Let $A$ be a lattice with corresponding language $\mathcal{L}$, $F \subset A$ be a set of designated values, and $M = \pair{W, R, v}$ be any $A$-model where $v$ is a normal modal valuation. For $w \in W$ and $\varphi$ in $\mathcal{L}$, we say that $w \Vdash_{A, F} \varphi$ if, and only if, $v_w(\varphi) \in F$.
    We call $\Vdash_{A, F}$ the \textbf{normal modal extension} of the matrix $H = \pair{A, F}$.
\end{definition}

If there are other operations in the lattice (e.g., for implication or for a different kind of negation), a normal modal valuation applies the Definition~\ref{normal-modal-valuation} for the operations $\land, \lor, \lnot$ and the corresponding operation for the additional operation. In this case, if $P: D^n \to D$ represents an operation in the lattice for the symbol $p$, the normal modal valuation will be such that $v(w, p(\varphi_1, \varphi_2, \ldots, \varphi_n)) = P(v(w, \varphi_1), v(w, \varphi_2), \ldots, v(w, \varphi_n))$.

\begin{theorem}\label{filter-and-regular}
For a complete lattice $A$ and a set $F$ of designate values in $A$, the $A$-normal modal valuation produces a unique regular modal extension of $\pair{A, F}$ if, and only if, $\bigwedge F \in F$ and $F$ is a filter.
\end{theorem}

\begin{proof}
Let us suppose $F$ is a filter in $A$ such that $\bigwedge F \in F$ and consider a $A$-model $M = \pair{W, R, v}$.  

If we assume that there is a world $w \in W$ satisfying $\Box \varphi$, we obtain 
$$\bigwedge \{v_{w'}(\varphi) \mid w R w'\} = f \in F$$
Therefore $v_{w'}(\varphi) \geq f$ for every $w'$ accessible to $w$. Since $F$ is a filter, $v_{w'}(\varphi) \in F$. This in turn means that each $w'$ accessible to $w$ satisfies $\varphi$.

If we assume that all $w'$ accessible to $w$ satisfies $\varphi$, we obtain $\{v_{w'}(\varphi) \mid w R w'\} \subseteq F$. It follows that $\bigwedge \{v_{w'}(\varphi) \mid w R w'\} \geq \bigwedge F$ and thus $v_w(\Box \varphi) \in F$ (i.e. $w \Vdash_A \Box \varphi$). 

Now, let us assume that $F$ is not a filter and there is $a \geq b \in F$ such that $a \notin F$. Then we consider a model with two worlds $w$ and $w'$ such that (i) $w R w$ and $w R w'$ and (ii) $v_w(C) = b$ and $v_{w'}(C) = a$ for a propositional variable $C$. We obtain $v_w(\Box C) = b$, and thus $w \Vdash_A C$. However, note that $w' \nVdash_A C$ since $a \notin F$. 

If, on the other hand, $\bigwedge F \notin F$, then we may obtain a model in which there is a world $w$ that accesses worlds $w_f$, $f \in F$, and such that for a propositional variable $C$ we have $v_{w_f}(C) = f$ for every $f \in F$. Therefore, we find that every world accessed by $w$ the formula $C$ is valid, but that $w$ does not validate $\Box C$.

\end{proof}

Let us briefly consider the finite and infinite cases for lattices and how it relates to regular modal valuations. Indeed, every finite lattice is complete. Additionally, whenever the set of designated values $F$ is a filter on the finite lattice $A$, we have $\bigwedge F \in F$. Hence, for a finite lattice, we can state \cref{filter-and-regular} simply as: 

\begin{corollary}
Let $A$ be a finite grid, and $F$ be a set of designated values in $A$. The normal extension $\Vdash_{A, F}$ is regular if, and only if, $F$ is a filter. 
\end{corollary}

For infinite cases, we have different scenarios:
\begin{enumerate}
    \item Consider the lattice $D$ with domain in the real line interval $[0,1]$ with the usual order and with negation $-x = 1-x$. The normal modal valuation will produce a regular modal extension when the filter is $F_1 = [0.5, 1]$ and a non-regular modal extension when it is $F_2 = ]0.5, 1]$. Note that $\bigwedge F_1 = 0.5 \in F_1$ and $\bigwedge F_2 = 0.5 \notin F_2$.
    \item Moreover, consider the Boolean algebra $B$ with the usual operations. Let $U$ be an ultrafilter over $B$. Consequently, the logic operating in the normal modal extension will be classical logic. Note that every finite model will operate precisely in the same way as in classical Kripke finite structures. However, if the ultrafilter $U$ is non-principal, then $\bigwedge U \notin U$. Consequently, the normal modal extension obtained will be regular if, and only if, the ultrafilter is principal.
\end{enumerate}
This shows that we can produce natural instances where NAW fails as $\pair{D, F_2}$ and $\pair{B, U}$ are usual lattice semantics for fuzzy logics and classical logic, respectively. 

\section{Formulas characterizing lattice structures}

Unless otherwise stated, we assume that implication $x \supset y$ is defined with the standard operation $-x+y$. However, it is important to note that this assumption is not favored in many important cases. Indeed, the logic LP (see. \cite[p. 142-160]{priest2008introduction}) defines implication as $-x+y$ and this choice undermines \textit{Modus Ponens}. However, a modification of implication can make \textit{Modus Ponens} valid in LP's three-valued lattice. One such change is simply to attribute top value\footnote{If the lattice has a top value, we will generally refer to this value with $1$ and, if it has a bottom value, we will generally refer to this value with $0$.} whenever the consequent is bigger than the antecedent and attribute the value of the consequent otherwise:

\begin{equation}\label{implicative-implication}
v(X \supset Y) = 
	\begin{cases}
		1&, v(X) \leq v(Y)\\
		v(Y)&, \text{ otherwise}
	\end{cases}
\end{equation}

This implication can be defined in any lattice that has a top value $1$. Implication is defined in this way for many non-classical logics to guarantee that the resulting logic satisfies the deduction theorem. Observe further that having an implication as in Equation~\ref{implicative-implication} can be very useful since in this case any set of designated values will be implicative. 

\begin{theorem}
    If $L$ is a lattice with top value $1$, implication $\supset$ as in Equation~\ref{implicative-implication} and $F$ is a set of designated values in $L$, then $F$ is implicative in $L$.
\end{theorem}

\begin{proof}
    Let $a$ and $b$ be values in $L$ such that $a \leq b$. Then $a \supset b = 1$. Because $F$ is upward closed, $1 \in F$ and so $(a \supset b) \in F$.
\end{proof}

Implicative filters preserve some important traditional modal validities in their normal modal extensions, as we see next. In short, we can show any implication by establishing that the antecedent is smaller than the consequent --- which may be very convenient in lattice semantics. Recall that most notable modal properties are written in terms of implications such as axioms K, 4, 5 and so on.

\begin{notation}
Let $X$ and $Y$ be subsets of the lattice $L$, we use $X + Y$ to refer to $\{x+y \mid x \in X \land y \in Y\}$.
\end{notation}

\begin{lemma} \label{indexed-distribution}
Let $L$ be a lattice, $F$ and $G$ be subsets of $L$ such that $\bigwedge F$, $\bigwedge G$ and $\bigwedge (F + G)$ exist, then
$$\bigwedge F + \bigwedge G \leq \bigwedge (F + G)$$
\end{lemma}

\begin{proof}
Fix $a = \bigwedge(F + G)$, $b = \bigwedge F$ and $c = \bigwedge G$. Indeed, every $x \in (F+G)$ is bigger than $b$, since $b$ is smaller than every member of $F$. Thus, from the maximality of $a$, $b \leq a$ and similarly $c \leq a$. Thence $b + c \leq a$ as desired.
\end{proof}

\begin{lemma}\label{distribution-disjunction}
Let $L = \pair{D, +, ., -, \supset}$ be a complete lattice, and $F$ be an implicative filter. Then for any $L$-model $M = \pair{W, R, v}$ and $w\in W$
$$w \Vdash_{L,F} (\Box \varphi \lor \Box \psi) \so \Box(\varphi \lor \psi)$$.
\end{lemma}

\begin{proof}
Let $w$ be a world in a $L$-model $M = \pair{W, R, v}$. We note that
\begin{equation*}
v_w(\Box \varphi \lor \Box \psi) = 
\bigwedge\underbrace{\{(v_{w'}(\varphi)) \mid w R w'\}}_{F} +
\bigwedge\underbrace{\{(v_{w'}(\psi)) \mid w R w'\}}_{G}
\end{equation*}
So, from Lemma~\ref{indexed-distribution}, we obtain that 
\begin{equation*}
v_w(\Box \varphi \lor \Box \psi) = \bigwedge F + \bigwedge G \leq \bigwedge (F + G)
\end{equation*}
But $\bigwedge (F + G) = \bigwedge \{(v_{w'}(\varphi) + v_{w'}(\psi)) \mid w R w'\}$. Thence, we have $v_w(\Box \varphi \lor \Box \psi) \leq v_w(\Box(\varphi \lor \psi))$.
Since $F$ is implicative, $w$ satisfies the formula $(\Box \varphi \lor \Box \psi) \so \Box(\varphi \lor \psi)$.
\end{proof}

We extend the property for models and frames using the usual strategy:

\begin{definition}
Let $A$ be a lattice, $F \subset A$ and $\Vdash$ be a modal extension of $\vDash_{A,F}$ in the language $L$. For a model $M = \pair{W, R, v}$, we will say that $M$ satisfies the $L$-formula $\varphi$ (i.e. $M \Vdash \varphi$) when, for every $w \in W$, $w \Vdash \varphi$.
\end{definition}

\begin{definition}
Let $A$ be a lattice, $F \subset A$ and $\Vdash$ be a modal extension of $\vDash_{A,F}$ in the language $L$. For a frame $\Fcal = \pair{W, R}$, we will say that $\Fcal$ satisfies the $L$-formula $\varphi$ (i.e. $\Fcal \Vdash \varphi$) when, for every $A$-model $M$ with frame $\Fcal$, $M \Vdash \varphi$.
\end{definition}

\begin{theorem}
Let $L$ be a complete lattice and $F$ be a set of designated values in $L$. Then $F$ is implicative if, and only if, every frame $\Fcal = \pair{W, R}$ is such that $\Fcal \Vdash_{L,F} (\Box \varphi \lor \Box \psi) \to \Box(\varphi \lor \psi)$.
\end{theorem}

\begin{proof}
That every $\Fcal$ satisfies $(\Box \varphi \lor \Box \psi) \to \Box(\varphi \lor \psi)$ under the designated values in $F$ comes directly from the Lemma~\ref{distribution-disjunction}.

Now suppose $F$ is not implicative. Then there are $a, b \in L$ such that $a \leq b$ and $a \supset b \notin F$. Consider a model $M$ with at least the worlds $w_1, w_2, w_3$ such that $w_1 R w_2$ and $w_1 R w_3$ and valuation $v_{w_2}(\varphi) = a$, $v_{w_2}(\psi) = b$, $v_{w_3}(\varphi) = b$, $v_{w_3}(\psi) = a$.  Then $v_{w_1}(\Box \varphi) = a$ and $v_{w_1}(\Box \psi) = a$; consequently $v_{w_1}(\Box \varphi \lor \Box \psi) = a$. On the other hand, $v_{w_1}(\Box (\varphi \lor \psi)) = b$. So $v_{w_1}((\Box \varphi \lor \Box \psi) \to \Box(\varphi \lor \psi)) \notin F$ as wanted.
\end{proof}

For dealing with axiom K, it will be instructive to consider different versions of implication. We will first consider the case that requires that implication satisfies a condition similar to that in Equation~\ref{implicative-implication} and then we consider the case where $x \supset y$ is simply $-x + y$.

\begin{definition}
Let $L$ be a lattice and let $F \subset L$ be a set of designated values. Define implication ($\supset$) as any function $L \times L \to L$. We say that implication is \textbf{deductive} if for every $a, b \in A$
\begin{enumerate}
	\item if $a \leq b$, then $b \leq (a \supset b) \in F$.
	\item if $a \nleq b$, then $(a \supset b) = b$.
\end{enumerate}
We say that implication is \textbf{strictly deductive} if there is a top value $1$ in $L$ and $(a \supset b) = 1$ when $a \leq b$.
\end{definition}

\begin{definition}
    We say that a lattice $L$ with filter $F$ is \textbf{linear outside} $F$ if every non-linearity of $L$ is contained in $F$, i.e., if $x \in L$ is not comparable with some $y \in L$, then $x \in F$.
\end{definition}

\begin{lemma}\label{lemma-for-K-implicative-implication}
Let $M$ be a $L$-model, $F$ be a filter in $L$ and say that implication is deductive in $L$. 
If $L$ is a linear lattice outside the filter, then every $w \in M$ satisfies ($\Vdash_{L,F}$) the axiom K, i.e. $\Box(\varphi \so \psi) \so (\Box \varphi \so \Box \psi)$.
\end{lemma}

\begin{proof}
Suppose $w \in M$ is such that $w \nVdash_F \Box(\varphi \so \psi) \so (\Box \varphi \so \Box \psi)$. Because implication is deductive, $a = v_w(\Box(\varphi \so \psi)) \nleq v_w(\Box \varphi \so \Box \psi) = c$ and $c \notin F$ for some values $a$ and $c$ in $L$. Subsequently, we have that $v_w(\Box \varphi) \nleq v_w(\Box \psi)$, $v_w(\Box \psi) \notin F$ and $v_w(\Box \psi) = c$. Consider  $A = \{v_x(\varphi \so \psi) \mid w R x\}$, $B = \{v_x(\varphi) \mid w R x\}$ and $C = \{v_x(\psi) \mid w R x\}$. 

Call $B_i$ and $C_i$ the elements of $B$ and $C$ such that $v_x(\varphi) \nleq v_x(\psi)$ and call $B_j$ and $C_j$ the elements of $B$ and $C$ such that $v_x(\varphi) \leq v_x(\psi)$. 
Recall that $\bigwedge C = c \notin F$ and $L$ is linear outside $F$.

Let $c_i = \bigwedge C_i$ and $c_j = \bigwedge C_j$. Since $c = c_i . c_j$, $c \notin F$ and $F$ is a filter, at least one of $c_i$ and $c_j$ is outside the filter. Because the lattice is linear outside the filter, we have $c_i = c$ or $c_j = c$.
Suppose $c_i = c$. In this case, since every element of $A$ is either in $C_i$ or in the filter and $c_i$ is smaller than every member of the filter, we conclude $\bigwedge A = a = c$. This in turn contradicts $v_w(\Box(\varphi \so \psi)) \nleq v_w(\Box \varphi \so \Box \psi)$. So we have $c_j = c$. However, note that $\bigwedge B_j \leq \bigwedge C_j$ once every member of $C_j$ has a smaller element in $B_j$. Trivially, $b \leq \bigwedge B_j$; so $b \leq c$, contradicting $v_w(\Box \varphi) \nleq v_w(\Box \psi)$.
\end{proof}

\begin{theorem}\label{thlinearoutsidethefilter}
Let $L$ be a complete lattice, $F$ be a non-empty set of designated values in $L$ and implication be strictly deductive in $L$. Then $L$ is linear outside $F$ if, and only if, every frame $\Fcal = \pair{W, R}$ is such that $\Fcal \Vdash_{L,F} \Box(\varphi \so \psi) \so (\Box \varphi \so \Box \psi)$.
\end{theorem}

\begin{proof}
If $L$ is linear outside the filter, then Lemma~\ref{lemma-for-K-implicative-implication} says that any frame satisfies axiom K. We then assume that $L$ is not linear outside the filter. Thus, there are $a, b \in L \setminus F$ such that $a \nleq b$ and $b \nleq a$. Consider a frame $\Fcal = \pair{W, R}$ with worlds $w_1$, $w_2$ and $w_3$ such that $w_1 R w_2$ and $w_1 R w_3$. Fix then a model $M$ with frame $\Fcal$ and such that $v_{w_2}(\varphi) = v_{w_3}(\varphi) = a$, $v_{w_2}(\psi) = b$ and $v_{w_3}(\psi) = a$. Now, observe that $v_{w_1}(\Box \psi) = a.b$ and $v_{w_1}(\Box \varphi) = a$. So we have $v_{w_1}(\Box \varphi \so \Box \psi) = a \supset a.b = a.b$ for $a \nleq a.b$. Moreover $v_{w_1}(\Box (\varphi \so \psi)) = (a \supset b) . (a \supset a) = b.1 = b$.
Hence $v_{w_1}(\Box(\varphi \so \psi) \so (\Box \varphi \so \Box \psi)) = (b \supset a.b) = a.b \notin F$.
\end{proof}

The converse in Theorem~\ref{thlinearoutsidethefilter} does not work without the assumption of strictly deductive. Possessing solely the implicative property, it is possible to identify a lattice $L$ that is non-linear outside the filter, where the reasoning used in proving Theorem~\ref{thlinearoutsidethefilter} is blocked. Consider the lattice $L$ where there is a pair of values $a \notin F$ and $b \notin F$ such that, for some value $f$, $a \supset a = f \in F$, $f.b \leq a.b$. 
In this case, we have $v_{w_1}(\Box \varphi \so \Box \psi) = a \supset a.b = a.b$ and $v_{w_1}(\Box (\varphi \so \psi)) = (a \supset b) . (a \supset a) = b.f$. This will give us $v_{w_1}(\Box (\varphi \so \psi) \so (\Box \varphi \so \Box \psi)) = (b.f \subset a.b) \in F$ since $L$ is implicative. It is possible to verify that a lattice consisting solely of the elements necessary for this example possesses the K property.\footnote{Precisely, the example will be the lattice $L =\{0, a, b, f, 1\}$, $(a \supset b) = f$ and $x \supset y$ is strictly deductive for all cases where $x \neq a$ and $y \neq b$. The order in $L$ is such that $1$ is top value, $0$ is bottom value and $a$, $b$ and $f$ are incomparable. All frames in this lattice with filter $\{f, 1\}$ satisfy axiom $K$. Examining this thoroughly is time-consuming and contributes little to the paper, so we leave it to the interested reader.}

Let us now consider the case where implication $x \supset y$ is defined as the regular $-x + y$.

\begin{lemma}\label{lemma-for-K-regular}
Let the lattice $L$ be anti-monotone with down-distribution, $M$ be a $L$-model and $F$ be a set of designated values in $L$. If $x \supset y$ is defined as $-x + y$ and $F$ is implicative, then every $w \in M$ satisfies ($\vDash_F$) the axiom K.
\end{lemma}

\begin{proof}
Consider any $w$ in $M$. Notably, $v_w(\Box \varphi) \leq v_{w'}(\varphi)$, for all $w R w'$. Thus, from anti-monotonicity, $- v_w(\Box \varphi) \geq -v_{w'}(\varphi)$, for all $w R w'$. Adding to both sides $v_{w'}(\psi)$, we obtain 
$$-v_w(\Box \varphi) + v_{w'}(\psi) \geq v_{w'}(\varphi \so \psi)$$
This being valid for every $w'$ accessed by $w$ give us
$$\bigwedge \{-v_w(\Box \varphi) + v_{w'}(\psi) \mid w R w'\} \geq v_w(\Box(\varphi \so \psi))$$
Because $L$ has down distribution,
$$-v_w(\Box \varphi) + \bigwedge \{v_{w'}(\psi) \mid w R w'\} \geq v_w(\Box(\varphi \so \psi))$$
$$v_w(\Box \varphi \so \Box \psi) \geq v_w(\Box(\varphi \so \psi))$$
As $L$ is implicative, $w \Vdash \Box(\varphi \so \psi) \so (\Box \varphi \so \Box \psi)$.
\end{proof}

\begin{theorem}\label{theorem-about-k-regular}
Let $L$ be anti-monotone, involutive and with down-distribution, $F$ be a set of designated values in $L$ and $x \so y$ is defined as $-x + y$ in $L$. Then $L$ is implicative if, and only if, every frame $\Fcal = \pair{W, R}$ is such that $\Fcal \Vdash_{L,F} \Box(\varphi \so \psi) \so (\Box \varphi \so \Box \psi)$.
\end{theorem}

\begin{proof}
From Lemma~\ref{lemma-for-K-regular}, we obtain that every frame satisfies $K$ if $L$ is implicative.
Let us then suppose that $L$ is not implicative. Then there are $a, b \in L$ such that $a \leq b$ and $a \so b \notin F$. 

Consider $w$ accessing $w_1$ and $w_2$ such that $v_{w_1}(\varphi) = -a$, $v_{w_2}(\varphi) = -b$, $v_{w_1}(\psi) = a$, $v_{w_2}(\psi) = a$. Naturally, $v_w(\Box \psi) = a$ and, since $L$ is anti-monotone, $v_w(\Box \varphi) = -b$. Hence $v_w(\Box \varphi \so \Box \psi) = --b + a = b$ because $L$ is involutive. Moreover, we have $v_w(\Box (\varphi \so \psi)) = a$. Since $a \so b \notin F$, $\Box(\varphi \so \psi) \so (\Box \varphi \so \Box \psi)$ is not valid in $w$.

\end{proof}

Without the anti-monotone or the down-distribution properties, we have too little control over the negation, and thus Lemma~\ref{lemma-for-K-regular} will fail. But even with those properties, the converse of Theorem~\ref{theorem-about-k-regular} fails without assuming that negation is involutive. The problem in this case is that we cannot use the double negation to extract the comparison we need like $v_w(\Box \varphi \so \Box \psi) = --b + a = b$. 

\section{An example}
    
From a Boolean algebra $B = \pair{D, +, ., -}$ one may obtain the twist algebra $T$ of $B$ as suggested in \cite{vakarelov1977notes, fidel1977algebraic}. The initial use for twist algebras was to provide semantics for Nelson logic \cite{nelson1949constructible}. It has recently been used to model some paraconsistent set theories \cite{Carnielli_Coniglio_2021} and used for modal systems in \cite{ono2014modal}. The twist $T = \pair{D \times D, +, ., -}$ is defined with respect to $B$ in such a way that for $a, b, c, d \in D$
\begin{enumerate}
    \item $\pair{a, b} + \pair{c, d} = \pair{a + c, b.d}$.
    \item $\pair{a, b} . \pair{c, d} = \pair{a . c, b + d}$.
    \item $ - \pair{a, b} = \pair{b, a}$.
\end{enumerate}
The twist lattice has anti-monotone and involutive properties. In fact $--\pair{a, b} = - \pair{b, a} = \pair{a, b}$. Additionally, if $\pair{a, b} \leq \pair{c, d}$, then $a \leq c$ and $b \geq d$. So $-\pair{a, b} = \pair{b, a} \geq \pair{d, c} = -\pair{c, d}$. Note further that any restriction of $T$'s domain will still have these properties so long as the restriction is closed by the operations $\{+, ., -\}$. Hence we have, for instance, anti-monotone and involutive properties for the restriction $P = \{\pair{a, b} \in D \times D \mid a+b = 1\}$; this is precisely the restriction used in \cite{Carnielli_Coniglio_2021} accounting for various paraconsistent systems. Notably, the restriction in $P$ is such that a twist value $t + (-t)$ is always of the form $\pair{1, x}$ while $t . (-t)$ is of the form $\pair{x, 1}$. Assuming a filter $F = \{\pair{1, x} \mid x \in D\}$ gives us a natural paraconsistent semantics. In this case, we should observe that the filter that gives us the axiom $K$ is precisely the one Carnielli and Coniglio use in \cite{Carnielli_Coniglio_2021}:

\begin{theorem}
Let $T$ be the twist algebra obtained from the Boolean algebra $B = \pair{D, +, ., -}$ and $P = \{\pair{a, b} \in D \times D \mid a+b = 1\}$. Let $F$ be a set of designated values in $P$, then $F \supset \{\pair{x, y} \mid x = 1\}$ if, and only if, every $P$-frame satisfies $K$.
\end{theorem}

\begin{proof}
Because $P$ is anti-monotone and involutive, Theorem~\ref{theorem-about-k-regular} give us that if every frame satisfies $K$, then $P$ is implicative in $F$. Note that, for any $x \in B$, $\pair{1, x} \geq \pair{1,1}$; so $(\pair{1, 1} \supset \pair{1,x}) = \pair{1, x} \in F$. On the other hand, suppose $F \supset \{\pair{x, y} \mid x = 1\}$. If $\pair{a, b} \leq \pair{c, d}$, then $a \leq d$, $b \geq c$ and $(\pair{a, b} \supset \pair{c, d}) = \pair{c + b, a.d}$. But, from $P$'s definition, $a + b = 1$ and $c + d =1$, thus we obtain $(\pair{a, b} \supset \pair{c, d}) = \pair{1, d} \in F$. Hence $P$ is implicative with filter $F$ and every $P$-frame satisfies K.
\end{proof}

\section{Final remarks}
The topic addressed in this article is broad and we have focused on a limited segment of modal logics based on lattice semantics.
We have studied some basic modal sentences that characterize algebraic properties of lattices. Indeed, in the presence of deductive implications, linearity is strictly connected to the validity of axiom K. Alternatively, for the regular implication defined semantically as $a \subset b = - a + b$, validity of K is connected to regularity properties for negation such as involutive and anti-monotone. 

One notable advantage of restricting our analysis to lattice-based structures is that the underlying order relation offers a natural foundation for interpreting the necessity operator. As emphasized in \cite{RM24}, this becomes especially relevant when considering models in which different worlds operate according to distinct logics. A shared order relation allows for meaningful interaction across such worlds, facilitating cross-logical communication.

Finally, we have also expressed interest in investigating the conditions under which normal modal extensions exhibit the finite-model property. This feature is of particular significance for practical applications, such as the implementation of logical frameworks in computational systems.

\paragraph{Acknowledgments.}

M. Martins was partially supported by FCT – Fundação para a Ciência e a Tecnologia through project UIDB/04106/2025 at CIDMA.

\printbibliography

\newpage

\end{document}